\documentclass{style}
\usepackage{amsfonts}
\usepackage{amsmath}
\usepackage{amssymb}
\usepackage{graphicx, subfigure}
\usepackage[hyphens]{url}
\usepackage{hyperref}
\usepackage{xspace}

\newcommand{\R}{\mathbb{R}}

\title{The Shadows of a Cycle Cannot All Be Paths}

\author{Prosenjit Bose\thanks{School of Computer Science, Carleton University, Ottawa ON, Canada, \protect\url{jit@scs.carleton.ca}, \protect\url{jdecaruf@cg.scs.carleton.ca}.}
\and Jean-Lou De Carufel\footnotemark[1]
\and Michael G.~Dobbins
  \thanks{Department of Mathematical Sciences, Binghamton University,
      Binghamton, NY, USA.
    \protect\url{michaelgenedobbins@gmail.com}}
\and Heuna Kim
 \thanks{Department of Mathematics, Arbeitsgruppe Theoretische Informatik, Freie Universit\"at Berlin, Berlin Germany. 
 \protect\url{heunak@mi.fu-berlin.de}}
\and Giovanni Viglietta\thanks{School of Electrical Engineering and Computer Science, University of Ottawa, Ottawa ON, Canada, \protect\url{viglietta@gmail.com}.}}

\index{Bose, Prosenjit}
\index{De Carufel, Jean-Lou}
\index{Dobbins, Michael G.}
\index{Kim, Heuna}
\index{Viglietta, Giovanni}

\begin{document}
\thispagestyle{empty}
\maketitle

\begin{abstract}
A \emph{shadow} of a subset $S$ of Euclidean space is an orthogonal 
projection of $S$ into one of the coordinate hyperplanes. In this paper we show that it is not possible for all three shadows of a cycle (i.e., a simple closed curve) in $\R^3$ to be paths (i.e., simple open curves).

We also show two contrasting results: the three shadows of a path in $\R^3$ can all be cycles (although not all convex) and, for every $d\geq 1$, there exists a $d$-sphere embedded in $\R^{d+2}$ whose $d+2$ shadows have no holes (i.e., they deformation-retract onto a point).
\end{abstract}

\section{Introduction}\label{s:intro}
\emph{Oskar's maze}, named after the Dutch puzzle designer Oskar van Deventer, who invented it in 1983, is a mechanical puzzle consisting of a hollow cube and three mutually orthogonal rods joined at their centers (see Figure~\ref{fig1}). Each face of the cube has slits forming a maze, and the mazes on opposite faces are identical. Each rod is orthogonal to a pair of opposite faces, and it is able to slide in the slits, tracing out the maze. Hence, in order to move the rods around, one has to solve three mazes simultaneously.

\begin{figure}[h]
\centering
\includegraphics[width=0.65\linewidth]{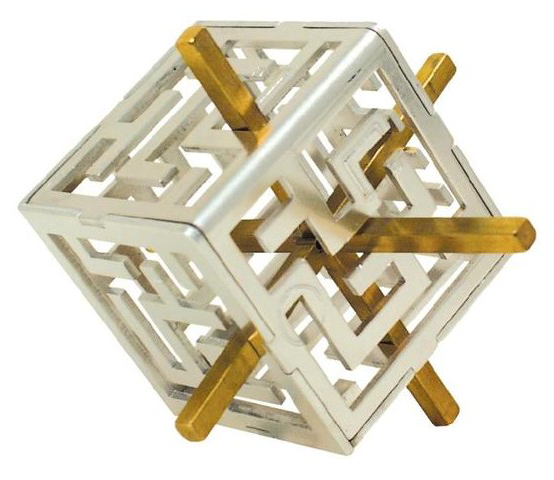}
\caption{Oskar's maze, produced by Bits and Pieces.}
\label{fig1}
\end{figure}

In 1994, Hendrik W.\ Lenstra asked if the mazes could be chosen so that the common point of the three rods could trace a simple closed curve.  Observe that none of the mazes may contain any cycles, or some pieces of the cube would fall out of the puzzle. So, what Lenstra was really asking for is a simple closed curve whose projections onto three pairwise orthogonal planes contain no cycles.  In other words, he wanted the three shadows of a simple closed 3D curve to all be trees.

As Peter Winkler reported in his book \emph{Mathematical mind-benders}~\cite{winkler}, a solution had already been found some years before by John R.\ Rickard, who discovered the curve illustrated in Figure~\ref{fig2}, also appearing on the front cover of Winkler's book.

Several other, more complex solutions to Lenstra's problem are known. Notably, in 2012 Adam P.\ Goucher constructed a simple closed curve having shadows that are all trees, which also happens to be a trefoil knot~\cite{blog}. The curve was therefore named \emph{Treefoil}. Goucher also constructed a pair of linked cycles whose union has shadows that are all trees.

\begin{figure}[h]
\centering
\includegraphics[width=0.8\linewidth]{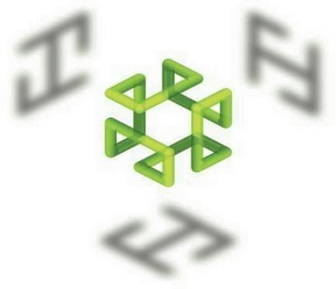}
\caption{Rickard's curve, illustrated by Afra Zomorodian, and appearing on the front cover of Peter Winkler's book \emph{Mathematical mind-benders}.}
\label{fig2}
\end{figure}

Our research is motivated by the following two questions. Is it possible for the three shadows of a simple closed curve to be paths, i.e., have neither cycles nor branch points? Can the three shadows of a simple open curve be simple closed curves? Both these questions are related to Lenstra's question, whose history is outlined in~\cite{winkler}. These questions have also been posed independently (see~\cite{cccg2007,workshop}).

\paragraph{Our contribution.}
In Section~\ref{s:paths} we answer the first question in the negative: the three shadows of a simple closed curve in $\R^3$ cannot all be paths. 

In Section~\ref{s:open} we answer the second question in the affirmative: there exist simple open curves in $\R^3$ whose three shadows are simple closed curves, although the shadows cannot all be convex. Furthermore, we exhibit a polygonal chain with this property having only six vertices, and we prove that six is the minimum.

In Section~\ref{s:higher} we extend Rickard's curve to higher dimensions, giving an inductive construction of a $d$-sphere embedded in $\R^{d+2}$, for every $d\geq 1$, whose $d+2$ shadows are all contractible. (A contractible set is one that can be continuously shrunk to a point, and hence it has no holes.)

Section~\ref{s:conclusions} concludes the paper with some remarks and suggestions for further work.

This research has obvious applications in computer vision and 3D object reconstruction, where the goal is to deduce properties of an unknown 3-dimensional object given its three projections. Specifically, we may want to study the topology of an object that projects to three given paths. It is easy to see that such an object may not be unique, and hence it makes sense to study the set of 3-dimensional objects that are \emph{compatible} with three given projections. Observe that any such set is closed under taking unions, and therefore it has a unique ``largest'' object, which is the union of all the objects in the set.

It is interesting to note that there are triplets of paths that are not compatible with any connected set, such as the one in Figure~\ref{fig6}. This means that an Oskar's-maze-like puzzle could be ``unsolvable'' even if it had no crossroads on any face. By ``unsolvable'' we mean that the set of locations that are reachable by the central point of the three rods depends on where the rods are located. Therefore, if we assign two points in the 3-dimensional maze determined by the three 2-dimensional mazes, it may be impossible to go from one to the other by moving the rods around.

\begin{figure}[h]
\centering
\includegraphics[scale=0.5]{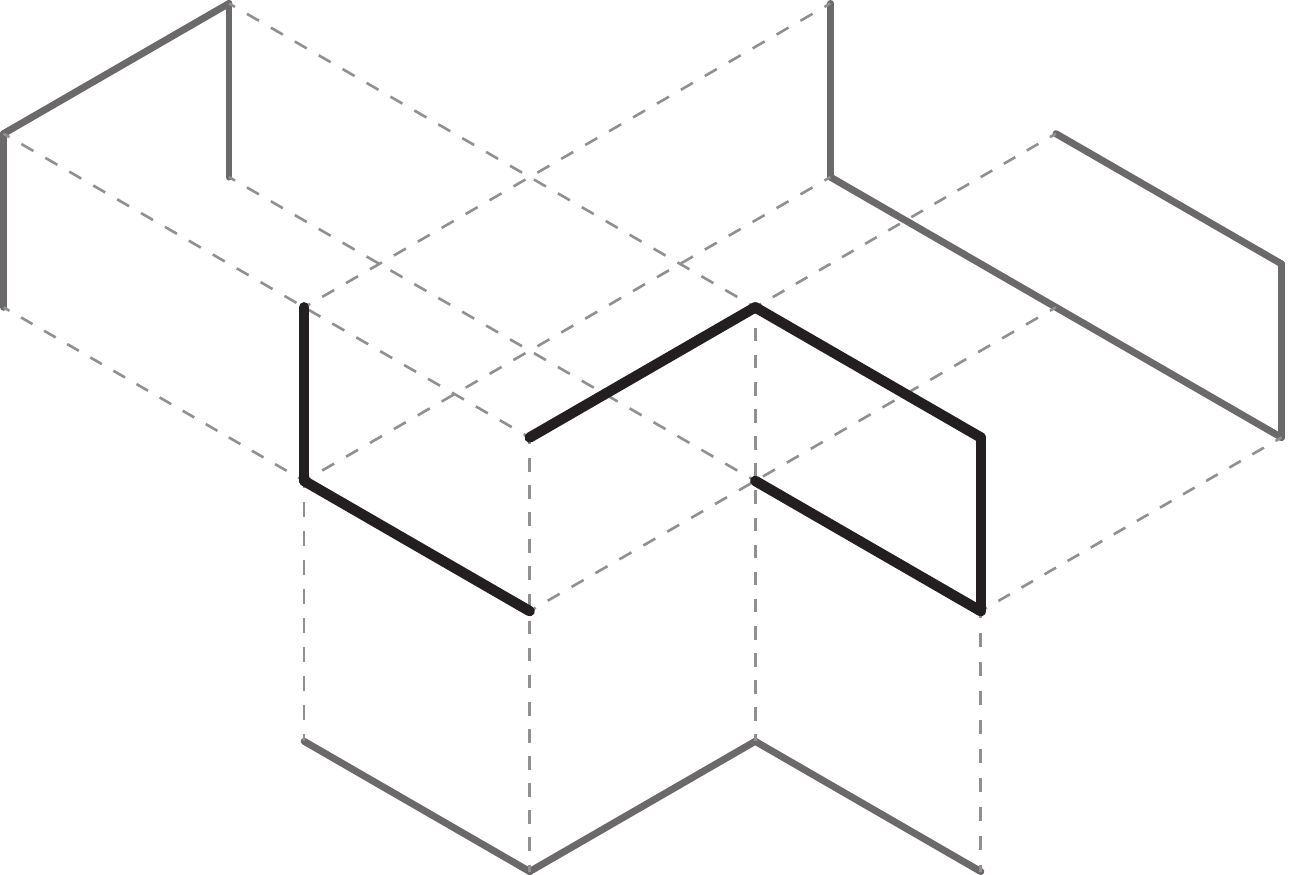}
\caption{Connected shadows whose unique compatible set is disconnected.}
\label{fig6}
\end{figure}

\section{The shadows of a cycle}\label{s:paths}
In this section we prove that the shadows of a simple closed curve in $\R^3$ cannot all be simple open curves. 
We start with some notation and definitions. 

For a point $p\in \R^n$ and $1\leq i\leq n$, we denote by $p_i$ the $i$-th coordinate of $p$. The \emph{$x_i$-projection}, or \emph{$x_i$-shadow}, of a set $A\subseteq\R^n$, denoted by $\pi_i(A)$, is the orthogonal projection of $A$ into the $i$-th coordinate hyperplane, e.g., $\pi_1(A) = \{(p_2,p_3,\cdots,p_n) \mid  p \in A \}$. If $A=\{p\}$, we may simply write $\pi_i(p)$ instead of $\pi_i(\{p\})$.

A \emph{path} is a (non-degenerate) simple open curve, and the \emph{interior} $\gamma^\circ$ of a path $\gamma$ is a copy of the path with its endpoints removed. A \emph{cycle} is a (non-degenerate) simple closed curve.

An \emph{$x_i$-strand} of a simple curve is a minimal path between the $x_i$-extremes of the curve.  
That is, an $x_i$-strand of a simple curve $\gamma$ with $x_i$-minimum $a_i = \min_{x \in \gamma} x_i$ and $x_i$-maximum $b_i = \max_{x \in \gamma} x_i$ is a path $\sigma \subseteq \gamma$ whose endpoints $s$ and $t$ are such that $s_i=a_i$ and $t_i=b_i$, and every internal point $x\in\sigma^\circ$ is such that $x_i \neq a_i,b_i$.

\begin{obs}\label{obs2}
The interiors of any two distinct $x_i$-strands of a simple curve are disjoint. Hence any two distinct $x_i$-strands of a path intersect at most at one common endpoint.
\end{obs}

\begin{obs}\label{obs1}
If $\sigma$ is an $x_i$-strand of a simple curve $\gamma$, then $\pi_j(\sigma)$ is an $x_i$-strand of $\pi_j(\gamma)$, for $j \neq i$.
\end{obs}
If $\pi_j(\gamma)$ is a path, the converse of Observation~\ref{obs1} is also true, as stated in the next lemma.

\begin{lemma}\label{lem_exstrand}
If $\sigma$ is an $x_i$-strand of the $x_j$-projection of a simple curve $\gamma$, with $i\neq j$, and $\pi_j(\gamma)$ is a path, then there exists an $x_i$-strand of $\gamma$ whose $x_j$-projection is $\sigma$.
\end{lemma}
\begin{proof}
Let $a$ and $b$ be the endpoints of $\pi_j(\gamma)$, let $a',b'\in\gamma$ such that $\pi_j(a')=a$ and $\pi_j(b')=b$, and let $\gamma'\subseteq\gamma$ be a path with endpoints $a'$ and $b'$. Since $\pi_j(\gamma)$ is a path, $\pi_j(\gamma)=\pi_j(\gamma')$. Let $c$ and $d$ be the endpoints of $\sigma$, such that $a$ and $c$ belong to the same connected component of $\pi_j(\gamma')\setminus \sigma^\circ$. Parameterizing $\gamma'$ from $a'$ to $b'$, let $c'$ be the last point of $\gamma'$ such that $\pi_j(c')=c$. Because $d$ separates $c$ and $b$ in $\pi_j(\gamma')$, there are points of $\gamma'$ after $c'$ whose $x_j$-projection is $d$. Letting $d'$ be the first of such points, the sub-path of $\gamma'$ with endpoints $c'$ and $d'$ is an $x_i$-strand of $\gamma$ whose $x_j$-projection is $\sigma$.
\end{proof}

In the following lemma we show that, if two shadows of a non-degenerate cycle are paths, then each of the two shadows has at least two similarly-oriented strands.

\begin{lemma}\label{lem_2strands}
If $\gamma$ is a cycle in $\R^3$ that is not contained in any $x_1$-orthogonal plane, and $\pi_2(\gamma)$ and $\pi_3(\gamma)$ are paths, then $\pi_3(\gamma)$ has at least two distinct $x_1$-strands. 
\end{lemma}
\begin{proof}
Since $\gamma$ is not in an $x_1$-orthogonal plane, $\gamma$ and $\pi_3(\gamma)$ both have at least one $x_1$-strand. 
Let $\sigma$ be an $x_1$-strand of $\gamma$, let $\tau_2=\pi_2(\sigma)$, and assume for contradiction that $\pi_3(\gamma)$ has a unique $x_1$-strand $\tau_3$, as sketched in Figure~\ref{figl1}.

\begin{figure}[h]
\centering
\includegraphics[scale=1.5]{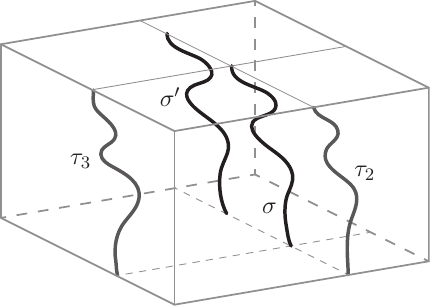}
\caption{Some $x_1$-strands of the curve in Lemma~\ref{lem_2strands}.}
\label{figl1}
\end{figure}

Since $\tau_2$ is an $x_1$-strand of $\pi_2(\gamma)$ by Observation~\ref{obs1}, and the endpoints of a path cannot be in the interior of one of its strands,
$\pi_2(\gamma) \setminus \tau_2^\circ$ contains the endpoints of $\pi_2(\gamma)$. Also, the $x_2$-shadow of $\gamma \setminus \sigma^\circ$ is a superset of $\pi_2(\gamma) \setminus \tau_2^\circ$, and hence it contains the endpoints of $\pi_2(\gamma)$, as well. Moreover, $\gamma \setminus \sigma^\circ$ is connected (because $\gamma$ is a cycle), hence $\pi_2(\gamma \setminus \sigma^\circ)$ is a connected subset of the path $\pi_2(\gamma)$ containing its endpoints, and therefore it must be all of $\pi_2(\gamma)$. By Lemma~\ref{lem_exstrand}, $\gamma$ has an $x_1$-strand $\sigma' \in \gamma \setminus \sigma^\circ$ such that $\pi_2(\sigma') = \pi_2(\sigma) = \tau_2$.
Since $\pi_3(\gamma)$ has a unique $x_1$-strand $\tau_3$, we have $\pi_3(\sigma') = \pi_3(\sigma) = \tau_3$, again by Observation~\ref{obs1}. 

Respectively parameterize $\sigma$, $\sigma'$, $\tau_2$, and $\tau_3$ each from the $x_1$-minimum $a_1$ to the $x_1$-maximum $b_1$ of $\gamma$.  Choose some value $c_1$ strictly between these extremes, $a_1 < c_1 < b_1$.  
Let $s \in \sigma$, $s' \in \sigma'$, $t_2 \in \tau_2$, $t_3 \in \tau_3$ respectively be the first point of each strand where the $x_1$ coordinate attains the value $c_1$. 
With this we have $\pi_2(s) = \pi_2(s') = t_2$ and $\pi_3(s) = \pi_3(s') = t_3$, which implies $s=s'$. So the interiors of the strands $\sigma$ and $\sigma'$ intersect, contradicting Observation~\ref{obs2}. Thus our assumption must be wrong: $\pi_3(\gamma)$ must have at least two distinct $x_1$-strands. 
\end{proof}

Next we prove that an $x_1$-strand and an $x_2$-strand of a planar path must intersect each other, and therefore their union must be a sub-path.

\begin{lemma}\label{lem_strandunion}
If $\sigma_1$ and $\sigma_2$ are respectively an $x_1$-strand and an $x_2$-strand of a path $\gamma$ in $\R^2$,
then $\sigma_1 \cup \sigma_2$ is a path.
\end{lemma}
\begin{proof}
Let $B=[a_1,b_1] \times [a_2,b_2]$ be the bounding box of $\gamma$, and let $s_1$ and $t_1$ be the leftmost and rightmost points of $\sigma_1$, respectively. Consider the polygonal chain $\tau$ with vertices $s_1$, $(a_1-1,a_2-1)$, $(b_1+1,a_2-1)$, $t_1$, in this order. Then $\sigma_1\cup\tau$ is a cycle which, by the Jordan Curve Theorem, disconnects the plane into two components: an interior $I$ and an exterior $E$.

Let $s_2$ and $t_2$ be the lowest and highest points of $\sigma_2$, respectively. Note that $s_2$ lies on the bottom edge of $B$, and hence it lies either in $I$ or on the curve $\sigma_1\cup\tau$. Similarly, $t_2$ lies on the top edge of $B$, and hence it lies either in $E$ or on the curve $\sigma_1\cup\tau$. Thus, $s_2\notin E$ and $t_2\notin I$. It follows that $\sigma_2$ must intersect $\R^2\setminus (I\cup E)=\sigma_1\cup \tau$. Since $\sigma_2\subset B$ and $\tau^\circ\cap B=\varnothing$, $\sigma_2$ must intersect $\sigma_1$.

Thus, $\sigma_1 \cup \sigma_2$ is a connected subset of the path $\gamma$, and is therefore a path.
\end{proof} 

In our final lemma we show that a planar path cannot have two distinct $x_1$-strands and two distinct $x_2$-strands.

\begin{lemma}\label{lem_path}
A path in $\R^2$ has either a unique $x_1$-strand or a unique $x_2$-strand. 
\end{lemma}
\begin{proof}
Assume for a contradiction that $\gamma$ is a path in $\R^2$ with distinct $x_1$-strands $\sigma_1, \sigma_2$ and distinct $x_2$-strands $\tau_1, \tau_2$. By Observation~\ref{obs2}, $\sigma_1$ and $\sigma_2$ are either disjoint, or their intersection is precisely a common endpoint. Suppose for a contradiction that they are disjoint, and let $\sigma'\subset \gamma$ be the minimal path connecting them. By Lemma~\ref{lem_strandunion}, $\sigma_1 \cup \tau_1$ is a path, as well as $\sigma_2 \cup \tau_1$, which implies that $\sigma'\subseteq \tau_1$. Similarly, $\sigma'\subseteq \tau_2$, and therefore $\sigma'\subseteq\tau_1\cap\tau_2$, contradicting Observation~\ref{obs2}. Thus $\sigma_1\cap\sigma_2$ is a single point $p$, and by a symmetric argument $\tau_1\cap\tau_2=p$, as well. Let $B = [a_1,b_1]\times [a_2,b_2]$ be the bounding box of $\gamma$. Then $p$ must be a vertex of $B$, and we may assume that $p=(b_1,b_2)$. Also, by symmetry, we may assume that $\tau_1 \subseteq \sigma_1$. It follows that $\sigma_1\cap\tau_2=\sigma_2\cap\tau_1=p$, and either $\tau_2 \subseteq \sigma_2$ or $\sigma_2 \subseteq \tau_2$.

\begin{figure}[h]
\centering
\subfigure[]{\label{figl2}\includegraphics[scale=0.525]{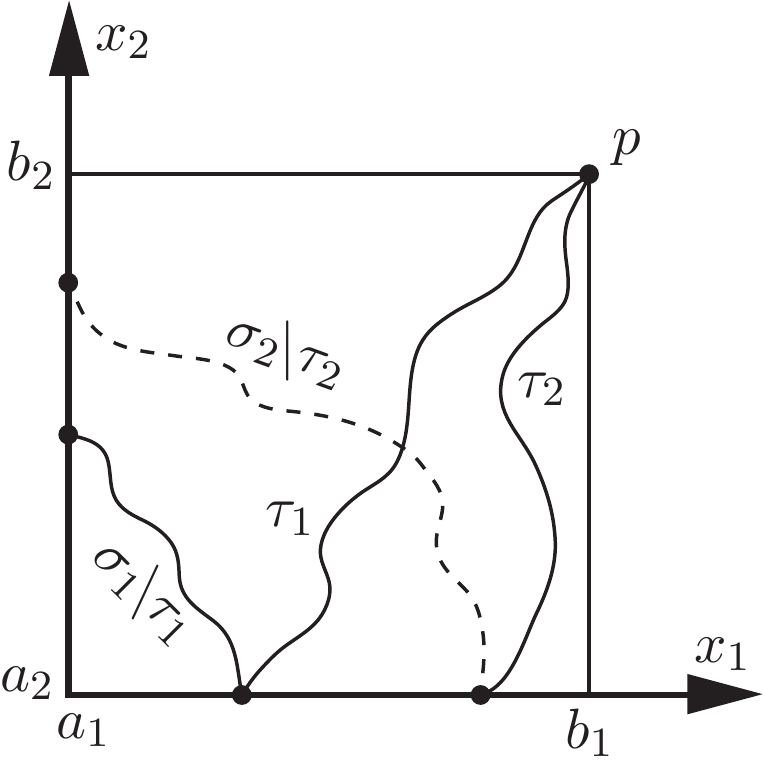}}\ \ \ 
\subfigure[]{\label{figl3}\includegraphics[scale=0.525]{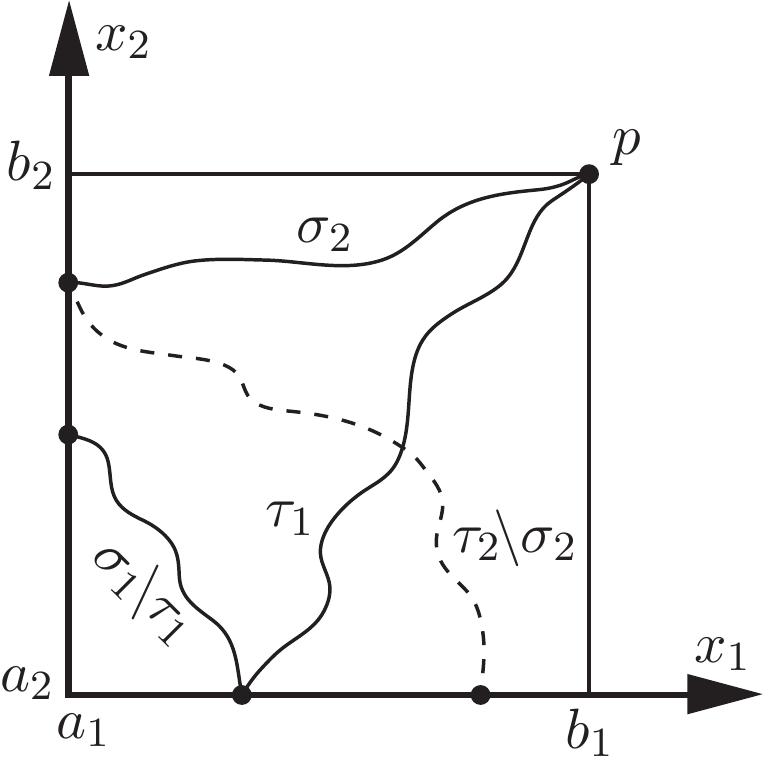}}
\caption{Cases of Lemma~\ref{lem_path}.}
\end{figure}

Suppose that $\tau_2 \subseteq \sigma_2$, as in Figure~\ref{figl2}.
Then $\tau_2^\circ$ is in the same connected component of $B \setminus \tau_1$ as the edge $\{b_1\} \times [a_2,b_2)$.
This implies that $\sigma_2 \setminus \tau_2$ intersects $\tau_1$, contradicting the fact that $\sigma_2\cap\tau_1=p$.

Suppose that $\sigma_2 \subseteq \tau_2$, as in Figure~\ref{figl3}.
Then $\sigma_2^\circ$ is in the same connected component of $B \setminus \sigma_1$ as the edge $[a_1,b_1)\times\{b_2\}$.
This implies that $\tau_2 \setminus \sigma_2$ intersects $\sigma_1$, contradicting the fact that $\sigma_1\cap\tau_2=p$. 

Thus our assumption fails: $\gamma$ has either a unique $x_1$-strand or a unique $x_2$-strand.
\end{proof}

We are now able to prove the main result of this section.

\begin{theorem}\label{thrm_shadowpath}
There is no cycle in $\R^3$ whose shadows are all paths.
\end{theorem}
\begin{proof}
Assume for a contradiction that the three shadows of a cycle $\gamma$ in $\R^3$ are all paths. 
Note that $\gamma$ cannot lie in any $x_i$-orthogonal plane, or $\pi_i(\gamma)$ would not be a path.
By Lemma \ref{lem_2strands}, since the $x_2$-shadow and the $x_3$-shadow are both paths, the $x_3$-shadow must have at least two distinct $x_1$-strands. 
Likewise, since the $x_1$-shadow and the $x_3$-shadow are both paths, the $x_3$-shadow must also have at least two distinct $x_2$-strands.  But by Lemma~\ref{lem_path}, a path in the $(x_1,x_2)$-plane cannot have two distinct $x_1$-strands and two distinct $x_2$-strands, which is a contradiction.
\end{proof}

\section{The shadows of a path}\label{s:open}
Here we study the simple open curves in $\R^3$ whose shadows are simple closed curves. In contrast with the similarly-defined curves of the previous section, in this case we can construct a wealth of such curves. An example is illustrated in Figure~\ref{fig7}.

\begin{figure}[h]
\centering
\includegraphics[scale=1]{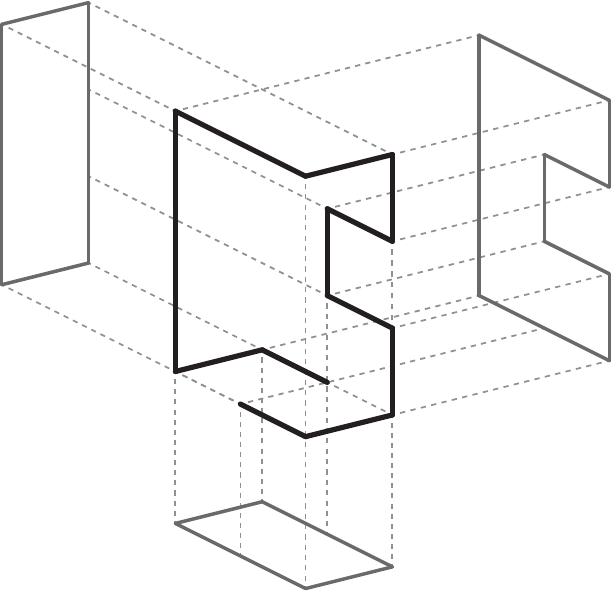}
\caption{Axis-aligned polygonal path whose shadows are all cycles.}
\label{fig7}
\end{figure}

Note that the curve in Figure~\ref{fig7} is a \emph{polygonal path} (i.e., a simple open polygonal chain) consisting of axis-parallel segments. If we allow arbitrarily oriented segments, we can find an example with only six vertices, which is the minimum possible.

\begin{theorem}
There exists a polygonal path in $\R^3$ with six vertices whose shadows are cycles. No such polygonal path exists with fewer than six vertices.
\end{theorem}
\begin{proof}
An example of such a polygonal path is $(1,0,1)\,(0,0,0)\,(1,1,0)\,(0,3,0)\,(2,0,2)\,(1,0,0)$, which is shown in Figure~\ref{fig8}.

\begin{figure}[h]
\centering
\includegraphics[scale=1]{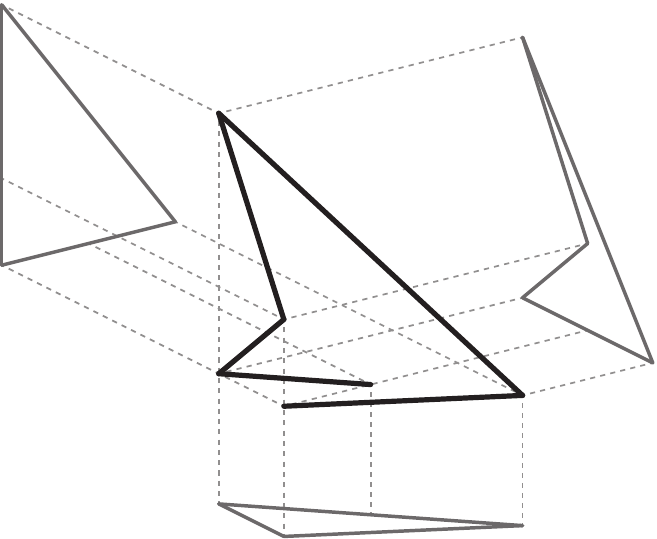}
\caption{Minimal polygonal path whose shadows are all cycles.}
\label{fig8}
\end{figure}

Suppose for a contradiction that a polygonal path in $\R^3$ with $n<6$ vertices exists such that its shadows are cycles. If $n\leq 3$, then clearly no shadow can be a cycle. Suppose that $n=4$, and let the polygonal path be $v_1\,v_2\,v_3\,v_4$. Then each shadow must be a triangle, and hence $\pi_i(v_1)=\pi_i(v_4)$ for every $i\in\{1,2,3\}$. It follows that $v_1=v_4$, which contradicts the fact that a polygonal path is an open curve.

Assume now that $n=5$, and the polygonal path is $v_1\,v_2\,v_3\,v_4\,v_5$. For every $i\in\{1,2,3\}$, the $x_1$-projection of the polygonal path is either a triangle or a quadrilateral. In both cases, the $x_i$-shadows of the segments $v_1\,v_2$ and $v_4\,v_5$ have a non-empty intersection. Since $v_1\neq v_5$, the $x_i$-shadows of $v_1$ and $v_5$ do not coincide for at least two $i$'s, say, $i=1$ and $i=2$. Then the $x_1$-shadow of the polygonal path must be a triangle, $\pi_i(v_1\,v_2)$ and $\pi_i(v_4\,v_5)$ are collinear, and hence the segments $v_1\,v_2$ and $v_4\,v_5$ lie on a plane that is orthogonal to the $(x_2,x_3)$-plane. Similarly, the segments $v_1\,v_2$ and $v_4\,v_5$ lie on a plane that is orthogonal to the $(x_1,x_3)$-plane, too. Hence $v_1\,v_2$ and $v_4\,v_5$ are either collinear or they lie on a common $x_3$-orthogonal plane. If $v_1\,v_2$ and $v_4\,v_5$ are collinear (and disjoint), then their $x_i$-shadows are disjoint for some $i\in\{1,2,3\}$, contradicting the fact that their intersection must be non-empty. If $v_1\,v_2$ and $v_4\,v_5$ lie on a common $x_3$-orthogonal plane, then $\pi_3(v_1\,v_2)$ and $\pi_3(v_4\,v_5)$ are disjoint, which is again a contradiction.
\end{proof}

Note that, in all the above examples, one of the shadows is a non-convex cycle. It is natural to ask whether a path exists whose shadows are all convex cycles. In the following theorem, we answer in the negative. (Due to space constraints, we only give a sketch of the proof.)

\begin{theorem}
There is no path in $\R^3$ whose shadows are convex cycles.
\end{theorem}
\textbf{Proof (sketch).}\ Suppose for contradiction that there exists a path $\gamma$ in $\R^3$ whose shadows are convex cycles. For every $i\in\{1,2,3\}$, $\gamma$ lies on the surface $\Gamma_i$ of a cylinder with section $\pi_i(\gamma)$ and $x_i$-parallel axis.

The intersection of $\Gamma_1$ and $\Gamma_2$ is sketched in Figure~\ref{figconv}. It consists of two horizontal axis-aligned rectangles $R_1$ and $R_2$ (assuming that the vertical direction is $x_3$-parallel) whose vertices are joined by four paths $\sigma_1$, $\sigma_2$, $\sigma_3$, and $\sigma_4$. The rectangles $R_1$ and $R_2$ may be degenerate, i.e., they may be $x_1$-parallel or $x_2$-parallel segments, or points. Let a horizontal plane intersect the interior of the path $\sigma_i$ in the point $s_i$, for each $i\in\{1,2,3,4\}$. Then, for all $i\in\{1,2,3,4\}$, either $s_i\in\gamma$ or $s_{i+1}\in\gamma$, where indices are taken modulo $4$.

\begin{figure}[h]
\centering
\includegraphics[scale=1]{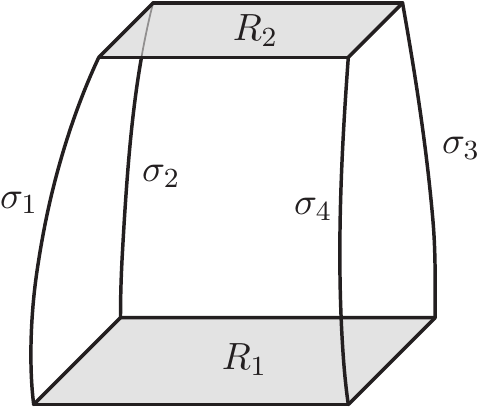}
\caption{Intersection of $\Gamma_1$ and $\Gamma_2$.}
\label{figconv}
\end{figure}

Further intersecting $\Gamma_1\cap\Gamma_2$ with $\Gamma_3$, we reduce $R_1$ and $R_2$ to at most four horizontal curves each. Therefore, in total we have $n\leq 12$ curves, whose union is an embedding in $\R^3$ of a graph $G$ with $n$ edges. Also, we may assume without loss of generality that each endpoint of $\gamma$ lies at a vertex of the embedding of $G$, or at the midpoint of one of the $n$ edges. Hence there are only finitely many possible graphs $G$ to consider, and only finitely many choices of $\gamma$ in each graph embedding. By exhaustively examining all the possible choices of $\gamma$, we conclude that none of them has shadows that are all convex cycles.\hfill$\square$

\section{Shadows in higher dimensions}\label{s:higher}
In this section we generalize Rickard's curve to higher dimensions. We inductively construct an embedding of a $d$-sphere in $\R^{d+2}$ whose $d+2$ shadows are all contractible, i.e., they deformation-retract to a point.

An \emph{$x_i$-slice} of a set $A\subseteq \R^n$, with $1\leq i\leq n$, is a non-empty intersection between $A$ and an $x_i$-orthogonal hyperplane.

\begin{theorem}
For every $d\geq 1$, there exists an embedding of a $d$-sphere in $\R^{d+2}$ whose shadows are all contractible.
\end{theorem}
\begin{proof}
Let $S_1$ be Rickard's curve, introduced in Section~\ref{s:intro}. Then, for all $d\geq 1$, we inductively define
$$S_{d+1}=\bigcup_{\lambda\in [-1,1]}(1-|\lambda|)\cdot S_d\times \{\lambda\}.$$
It is easy to see that $S_d$ is an embedding of a $d$-sphere in $\R^{d+2}$ for every $d\geq 1$. We claim that all the shadows of $S_d$ deformation-retract to the point $\{0\}^{d+1}$. This is true for $d=1$, as suggested by Figure~\ref{fig2}. Assume now the inductive hypothesis that the claim is true for $S_d$, and therefore there exists a continuous map
$$F_{d,i}\colon \pi_i(S_d)\times [0,1]\to \pi_i(S_d)$$
with $F_{d,i}(x,0)=x$ and $F_{d,i}(x,1)=\{0\}^{d+1}$, for every $1\leq i\leq d+2$. Now, for each $1\leq i\leq d+3$, we can construct a continuous map
$$F_{d+1,i}\colon \pi_i(S_{d+1})\times [0,1]\to \pi_i(S_{d+1})$$
with $F_{d+1,i}(x,0)=x$ and $F_{d+1,i}(x,1)=\{0\}^{d+2}$.

\begin{figure}[h]
\centering
\includegraphics[scale=1]{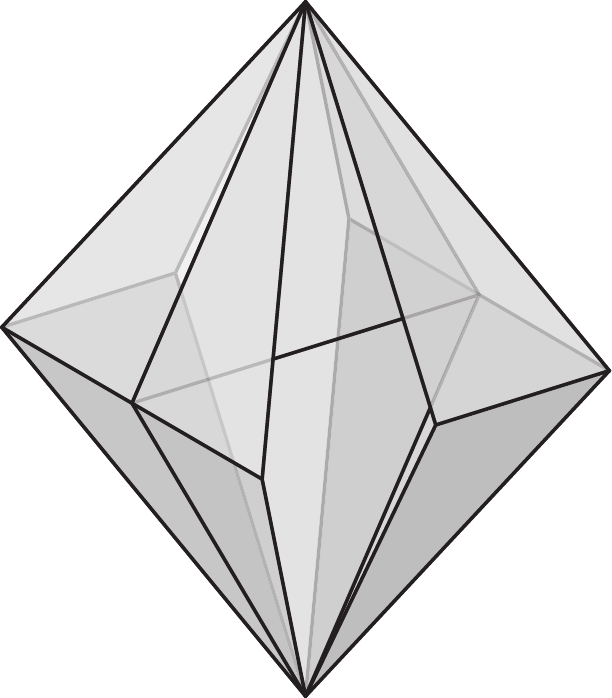}
\caption{$x_i$-shadow of $S_2$, for $1\leq i\leq 3$.}
\label{figh1}
\end{figure}

If $1\leq i\leq d+2$, we first define the auxiliary map
$$F'\colon \pi_i(S_{d+1})\times [0,1]\to \pi_i(S_{d+1})$$
as follows. For every $x\in \pi_i(S_{d+1})$ such that $|x_{d+2}|\neq 1$ and $\lambda\in [0,1]$, we let
$$F'(x,\lambda)=\left(1-|x_{d+2}|\right)\cdot F_{d,i}\left(\frac{\pi_{d+2}(x)}{1-|x_{d+2}|},\lambda\right)\\
\times\left\{x_{d+2}\right\}.$$
If $x\in \pi_i(S_{d+1})$ with $|x_{d+2}|=1$ and $\lambda\in [0,1]$, we let $F'(x,\lambda)=x.$ Observe that every $x_{d+2}$-slice of $\pi_i(S_{d+1})$ is a scaled copy of $\pi_i(S_d)$. (Figure~\ref{figh1} shows $\pi_i(S_{d+1})$ for $d=1$.) Informally, $F'$ applies $F_{d,i}$ with parameter $\lambda$ to a suitably scaled copy of each $x_{d+2}$-slice, and then it rescales it back. Therefore, since $F_{d,i}$ is a deformation retraction of $\pi_i(S_d)$ to the point $\{0\}^{d+1}$, $F'$ is a deformation retraction of $\pi_i(S_{d+1})$ to the segment $\{0\}^{d+1}\times [-1,1]$. To obtain $F_{d+1,i}$, one just has to compose $F'$ with a deformation retraction of $\{0\}^{d+1}\times [-1,1]$ to the point $\{0\}^{d+2}$. In formulas, for $x\in \pi_i(S_{d+1})$ and $\lambda\in [0,1]$,
$$F_{d+1,i}(x,\lambda) = \left\{
  \begin{array}{ll}
    F'(x,2\lambda) & \mbox{if\ } \lambda<1/2\\
    (2-2\lambda)\cdot F'(x,1) & \mbox{if\ } \lambda\geq 1/2.
  \end{array}
\right.$$

\begin{figure}[h]
\centering
\includegraphics[scale=1]{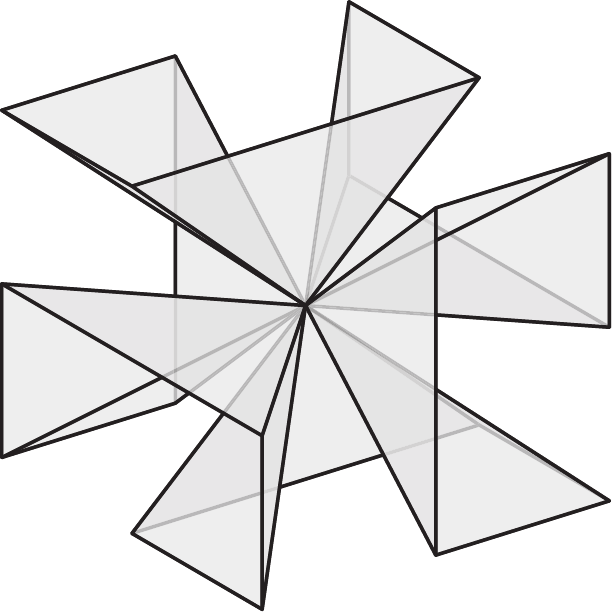}
\caption{$x_4$-shadow of $S_2$.}
\label{figh2}
\end{figure}

If $i=d+3$, we can simply set
$$F_{d+1,i}(x,\lambda)=(1-\lambda)\cdot x$$
for every $x\in \pi_i(S_{d+1})$ and $\lambda\in [-1,1]$. This is easily seen to be a deformation retraction to $\{0\}^{d+2}$. (Figure~\ref{figh2} shows $\pi_i(S_{d+1})$ for $d=1$.)

Hence all the shadows of the $d$-sphere $S_d$ deformation-retract to a point for every $d\geq 1$, meaning that they are contractible.
\end{proof}

\section{Concluding remarks}\label{s:conclusions}
In this paper we studied the shadows of curves in $\R^3$ (a shadow being an axis-parallel projection), also settling some long-standing open problems posed in~\cite{cccg2007,workshop}.

In Section~\ref{s:paths} we proved that there is no cycle in $\R^3$ whose shadows are all paths. Note that by applying a projective transformation, we may equivalently define shadows to be perspective projections, provided that the three viewpoints are not collinear, and the plane through them does not intersect the curve.

In Section~\ref{s:open} we proved that there exist paths in $\R^3$ whose shadows are all cycles. We also showed that, if such a path is a polygonal chain, it must have at least six vertices, and we found an example with exactly six vertices. Then we proved that there is no path in $\R^3$ whose shadows are all convex cycles.

Finally, in Section~\ref{s:higher} we showed that there exists an embedding of a $d$-sphere in $\R^{d+2}$ whose shadows are all contractible, for every $d\geq 1$. This generalizes Rickard's curve (see Figure~\ref{fig2}), which is a cycle in $\R^3$ whose shadows contain no cycles.

Our results can be expanded in several directions. A natural goal would be to minimize the total number of branch points of the shadows of a cycle in $\R^3$, assuming that all shadows are cycle-free. Because each shadow of Rickard's curve has two branch points, such a minimum is at most six. On the other hand, by Theorem~\ref{thrm_shadowpath}, the minimum is at least one. With the same proof technique employed in Section~\ref{s:paths}, we can prove the following generalized version of Theorem~\ref{thrm_shadowpath}, which implies that the minimum number of branch points of the shadows must be at least three.

\begin{theorem}
There is no cycle $\gamma$ in $\R^3$ with cycle-free shadows such that $\pi_1(\gamma)$ is a path, and $\pi_2(\gamma)$ and $\pi_3(\gamma)$ have at most one branch point each.\hfill$\square$
\end{theorem}
We conjecture Rickard's curve to be an optimal example in terms of branch points of its shadows.

\begin{conj}
If the shadows of a cycle in $\R^3$ are all cycle-free, then each shadow has at least two branch points.
\end{conj}

\small{\paragraph{Acknowledgments.}
The authors are indebted to Giuseppe Antonio Di Luna, Pat Morin, and Joseph O'Rourke for stimulating discussions.

Prosenjit Bose and Jean-Lou De Carufel were supported in part by NSERC.

Michael G.~Dobbins was supported by NRF grant 2011-0030044 (SRC-GAIA)
funded by the government of South Korea.

Heuna Kim was supported by the Deutsche Forschungsgemeinschaft within the research training group ``Methods for Discrete Structures'' (GRK 1408).}

\small
\bibliographystyle{abbrv}

\end{document}